\documentclass[letterpaper, 10 pt, conference]{ieeeconf}  

\IEEEoverridecommandlockouts                              
\usepackage[latin9]{inputenc}
\usepackage{mathtools}
\usepackage{amsmath}
\usepackage{amssymb}

\usepackage{amsthm}

\usepackage{graphicx}
\PassOptionsToPackage{normalem}{ulem}
\usepackage{ulem}
\usepackage{microtype}
\usepackage[unicode=true,
bookmarks=false,
breaklinks=false,pdfborder={0 0 1},backref=section,colorlinks=false]
{hyperref}
\usepackage{subfigure}

\usepackage{cleveref}

\newcommand{\ux}{\underline{x}}
\newcommand{\ox}{\overline{x}}

\global\long\def\X{\mathbb{X}}%
\global\long\def\U{\mathbb{U}}%

\renewcommand{\epsilon}{\varepsilon}
\renewcommand{\tilde}{\widetilde}
\renewcommand{\bar}{\overline}

\def\:#1{\protect \ifmmode {\mathbf{#1}} \else {\textbf{#1}} \fi}

\newcommand{\cA}{\mathcal{A}}
\newcommand{\cB}{\mathcal{B}}

\newcommand{\cD}{\mathcal{D}}

\newcommand{\cL}{\mathcal{L}}

\newcommand{\cU}{\mathcal{U}}

\renewcommand{\epsilon}{\varepsilon}

\newcommand{\transp}{\top}

\DeclarePairedDelimiterX{\inp}[2]{\langle}{\rangle}{#1, #2}

\newcommand{\Real}{\mathbb{R}}
\newcommand{\Natural}{\mathbb{N}}

\newtheorem{theorem}{Theorem}

\newtheorem{corollary}{Corollary}

\newtheorem{proposition}{Proposition}

\newtheorem{assumption}{Assumption}

\title{\LARGE \bf
	Robust-Adaptive Interval Predictive Control \\for Linear Uncertain Systems
}

\author{Edouard Leurent and Denis Efimov and Odalric-Ambrym Maillard
	\thanks{Edouard Leurent is with Renault Group, Paris, France.}%
	\thanks{Denis Efimov and Odalric-Ambrym Maillard are with Inria Valse and Inria SequeL respectively, Inria Lille Nord-Europe, France.}%
}

\begin{document}

\maketitle
\thispagestyle{empty}
\pagestyle{empty}

\begin{abstract}
We consider the problem of stabilization of a linear system, under
state and control constraints, and subject to bounded disturbances
and unknown parameters in the state matrix. First, using a simple
least square solution and available noisy measurements, the set of
admissible values for parameters is evaluated. Second, for the estimated
set of parameter values and the corresponding linear interval model
of the system, two interval predictors are recalled and an unconstrained
stabilizing control is designed that uses the predicted intervals.
Third, to guarantee the robust constraint satisfaction, a model predictive
control algorithm is developed, which is based on solution of an optimization
problem posed for the interval predictor. The conditions for recursive
feasibility and asymptotic performance are established. Efficiency
of the proposed control framework is illustrated by numeric simulations.
\end{abstract}

\section{Introduction}

There are plenty of real-world control problems for dynamical systems,
which face a severe model uncertainty (that can be represented by
unknown parameters and exogenous disturbances), under strict state
and control constraints, whose maintaining is critical and related
with the system safety (\emph{e.g}., path/trajectory planning for
autonomous cars and robots \cite{leurent2019interval,Lenz2015}).
The most popular approaches aiming to solve these complex regulation
issues are the methods based on reinforcement (deep) learning \cite{mnih2015humanlevel} 
or the Model Predictive Control (MPC) algorithms \cite{Basar1996,MPC1,MPC:Tube2},
which are more common in industrial applications (there are also techniques
relying on both frameworks as in \cite{Lenz2015,Dean2018}). The advantages
of these tools consist in the ability to provide a robust constraint
satisfaction, while ensuring optimization of a selected cost. A shortage
is their increased (online) computational complexity, which becomes
less important with growing numeric capacities of smart sensors and
actuators that are omnipresent nowadays.

Focusing on the latter group of approaches, the adaptive MPC is a
common framework to counteract the influence of uncertain parameters
\cite{Fukushima2007,Adetola2009,Adetola2011,Aswani2013,Vicente2019},
which includes, first, an estimation/adaptation algorithm to evaluate
the model uncertainty, and second, an MPC algorithm that helps to
keep the constraints during transients. If the constraint satisfaction
is predominant for system's safety, the complicacy comes from the
evaluation of all possible plant trajectories in the presence of external
perturbations and noises. 

In the present work, our goal is to develop a simple solution that
allows the system comportment to be evaluated for given model uncertainty
with adjustments provided by estimation algorithms. Such a solution
is based on interval predictors proposed recently \cite{Efimov2012,leurent2019interval},
whose use will provide to a slightly nonlinear MPC algorithm an ability
to keep the constraints in all admissible scenarios.

The outline of this work is as follows. The detailed problem statement
is given in Section \ref{sec:Problem}. We consider a continuous-time
linear system with uncertainty presented by state and output disturbances,
as well as a vector of unknown parameters belonging to a given compact
set; and our goal is the state stabilization while satisfying the
state and the control restrictions (due to presence of disturbances
the input-to-state stability concept is used). A simple parameter
estimation algorithm with evaluation of its accuracy is discussed
in \ref{sec:estimation} (this part is rather standard and does
not constitute the main novelty). The obtained set of admissible values
for the uncertain parameters allows an interval model to be obtained
for the considered system. Two interval predictors are introduced
in Section \ref{sec:prediction}, together with an unconstrained robust
stabilizing control that uses the interval predictor variables only.
In Section \ref{sec:control}, an MPC algorithm based on the designed interval predictors is developed, and analyzed in our main result. Applicability
of the approach is demonstrated on lane-keeping application for a self-driving car in Section \ref{sec:experiments}.

\section*{Notation}

Denote $[n]=\{1,2,\dots,n\}$ for any $n\in\Natural$. Euclidean norm is denoted
as $|\cdot|$, and $L_{\infty}$ norm on $[t_{0},t_{1})$ as $\Vert\cdot\Vert_{[t_{0},t_{1})}$.
We denote as $\mathcal{L}_{\infty}^{m}$ the set of all inputs
$u:\mathbb{R}_{+}\to\mathbb{R}^{m}$ with the property $\Vert u\Vert_{[0,+\infty]}<\infty$.
Given a matrix $A\in\Real^{n\times n}$, denote $A^{+}=\max\{A,0\}$,
$A^{-}=A^{+}-A$, $|A|=A^{+}+A^{-}$. 
For two vectors $x_{1},x_{2}\in\mathbb{R}^{n}$ or matrices $A_{1},A_{2}\in\Real^{n\times n}$,
the relations $x_{1}\le x_{2}$ and $A_{1}\le A_{2}$ are understood
elementwise. 
The relation $P\prec0$ ($P\succeq0$) means that a symmetric matrix
$P\in\Real^{n\times n}$ is negative (positive semi) definite.

\section{\label{sec:Problem} Problem Statement}

We consider a linear system: 
\begin{equation}
\dot{x}(t)=A(\theta)x(t)+Bu(t)+D\omega(t),\;t\geq0,\label{eq:dynamics}
\end{equation}
where $x(t)\in\Real^{p}$ is the state, $u(t)\in\Real^{q}$ is the
control and $\omega(t)\in\Real^{r}$ is the state perturbation, $\omega\in\cL_{\infty}^{r}$;
it is assumed that the constant uncertain parameter vector $\theta\in\Real^{d}$
in the state matrix $A:\Real^{d}\to\Real^{p\times p}$ belongs to
a compact set $\Theta\subset\Real^{d}$; the control matrix $B\in\Real^{p\times q}$
and disturbance matrix $D\in\Real^{p\times r}$ are known. We also
assume that the noisy observations of $x(t)$ and and $\dot{x}(t)$
are available:
\begin{gather}
y_{1}(t)=x(t)+\nu_{1}(t),\;y_{2}(t)=\dot{x}(t)+\nu_{2}(t),\label{eq:outputs}
\end{gather}
where $\nu(t)=[\nu_{1}^{\top}(t)\;\nu_{2}^{\top}(t)]^{\top}\in\Real^{2p}$
is the measurement disturbance, $\nu\in\cL_{\infty}^{2p}$. Roughly
speaking, we assume with \eqref{eq:outputs} that the state $x(t)$ and its derivative $\dot{x}(t)$
are estimated using some observation/differentiation approaches \cite{Besancon2007,Reichhartinger2018}.

\begin{assumption}
\label{assu:main} There exist signals $\underline{\omega},\overline{\omega}\in\cL_{\infty}^{r}$,
$\underline{\nu},\overline{\nu}\in\cL_{\infty}^{2p}$ and two vectors
$\underline{x}_{0},\overline{x}_{0}\in\Real^{p}$ such that
\begin{gather*}
\underline{\omega}(t)\leq\omega(t)\leq\overline{\omega}(t),\;\underline{\nu}(t)\leq\nu(t)\leq\overline{\nu}(t)\quad\forall t\geq0,\\
\underline{x}_{0}\leq x(0)\leq\overline{x}_{0}.
\end{gather*}
\end{assumption}

\subsection{Problem}

Our goal is to design a robust control that stabilizes \eqref{eq:dynamics},
\eqref{eq:outputs} at a vicinity of the origin under Assumption \ref{assu:main}
such that
\begin{equation}
x(t)\in\X,\;u(t)\in\U\quad\forall t\geq0,\label{eq:constraints}
\end{equation}
where $[\underline{x}_{0},\overline{x}_{0}]\subset\X\subset\Real^{p}$
and $\U\subset\Real^{q}$ are given bounded constraint sets for the
state and the control, respectively.

\section{\label{sec:estimation} Model Estimation}

To derive a confidence region $\hat{\Theta}(t)\subseteq\Theta$ for
the uncertain parameters $\theta$, the structure of
$A(\theta)$ must be specified:
\begin{assumption}
\label{assu:structure} There exist known matrices $A,\phi_{1},\dots,\phi_{d}\in\Real^{p\times p}$
such that for all $\theta\in\Theta$,
\[
A(\theta)=A+\sum_{i=1}^{d}\theta_{i}\phi_{i}.
\]

\end{assumption}
\eqref{eq:dynamics}, \eqref{eq:outputs} and \Cref{assu:structure}
yield the linear regression:
\begin{equation}
y(t)=\Phi(t)\theta+\eta(t),\label{eq:lin_regr}
\end{equation}
where $y(t)=y_{2}(t)-Ay_{1}(t)-Bu(t)$ and $\Phi(t)=[\phi_{1}y_{1}(t)\dots\phi_{d}y_{1}(t)]\in\Real^{p\times d}$
are known signals, and
\begin{align*}
\eta(t) & =D\omega(t)-\left(A+\textstyle\sum_{i=1}^{d}\theta_{i}\phi_{i}\right)\nu_{1}(t)+\nu_{2}(t)
\end{align*}
is the combined perturbation, which is from $\cL_{\infty}^{p}$ under
assumptions \ref{assu:main} and \ref{assu:structure}:
$
\Vert\eta\Vert_{\infty}\leq\bar{\eta}.
$

We need a hypothesis on the level of excitation of $\Phi$ \cite{Morgan1977}:
\begin{assumption}
\label{assu:PE} There exist $\ell>0$ and $\vartheta>0$ such that
the matrix function $\Phi:\mathbb{R}_{+}\to\mathbb{R}^{p\times d}$
satisfies $(\ell,\vartheta)$--Persistence of Excitation (PE) condition: for any $t\in\mathbb{R_{+}}$,
\[
\int_{t}^{t+\ell}\Phi^{\top}(s)\Phi(s)ds\ge\vartheta I_{d}.
\]

\end{assumption}
Under Assumption \ref{assu:PE}, for any $g\geq1$ and $g\ell<T\leq(g+1)\ell$,
\begin{align*}
\int_{t}^{t+T}\Phi^{\top}\Phi(s)ds = & \left(\int_{t+g\ell}^{t+T}+ \sum_{i=0}^{g-1}\int_{t+i\ell}^{t+(i+1)\ell}\right)\Phi^{\top}\Phi(s)ds\\
\geq & \frac{g}{g+1}\frac{\vartheta}{\ell}TI_{d}\geq\frac{\vartheta}{2\ell}TI_{d},
\end{align*}
\emph{i.e}., the matrix function $\int_{t}^{t+T}\Phi^{\top}(s)\Phi(s)ds$
is nonsingular.
Taking into account this observation, in order to solve \eqref{eq:lin_regr}
and obtain $\hat{\Theta}(t)$, we use next the simplest least
square estimation:
\begin{equation}
\hat{\theta}(t)=\begin{cases}
\left(\int_{0}^{t}\Phi^{\top}(s)\Phi(s)ds\right)^{-1}\int_{0}^{t}\Phi^{\top}(s)y(s)ds & t\geq\ell\\
\theta_{0} & t\in[0,\ell)
\end{cases},\label{eq:LS}
\end{equation}
where $\hat{\theta}(t)\in\Real^{d}$ is an estimate of $\theta$ and
$\theta_{0}\in\Theta$ is an initial estimate. 
The estimation error of \eqref{eq:LS} can be evaluated as: 
\begin{proposition}
\label{prop:LS} Let assumptions \ref{assu:main}, \ref{assu:structure}
and \ref{assu:PE} be satisfied. Then, for all $t\geq\ell$:
$
|\hat{\theta}(t)-\theta|\leq\Delta\theta(\Vert x\Vert_{\infty}),\text{ where}
$
\[
\Delta\theta(\Vert x\Vert_{\infty})=\frac{2\ell}{\vartheta}\max_{i\in[d]}\Vert\phi_{i}\Vert_{2}(\Vert x\Vert_{\infty}+\max\{\Vert\underline{\nu}_{1}\Vert_{\infty},\Vert\overline{\nu}_{1}\Vert_{\infty}\})\bar{\eta}.
\]
\end{proposition}
\begin{proof}
According to \eqref{eq:lin_regr}, the algorithm \eqref{eq:LS} can
be rewritten for $t\geq\ell$ in the form:
\[
\hat{\theta}(t)=\theta+\left(\int_{0}^{t}\Phi^{\top}(s)\Phi(s)ds\right)^{-1}\int_{0}^{t}\Phi^{\top}(s)\eta(s)ds.
\]
Using the fact that $\int_{0}^{t}\Phi^{\top}(s)\Phi(s)ds\geq\frac{\vartheta}{2\ell}tI_{d}$
for $t\geq\ell$, the claim follows the direct computations:
\begin{align*}
|\hat{\theta}(t)-\theta|&\leq\left\Vert \left(\int_{0}^{t}\Phi^{\top}(s)\Phi(s)ds\right)^{-1}\right\Vert _{2}\left|\int_{0}^{t}\Phi^{\top}(s)\eta(s)ds\right|\\
&\leq\frac{2\ell}{\vartheta t}\int_{0}^{t}\Vert\Phi(s)\Vert_{2}|\eta(s)|ds \leq \text{claimed bound}.
\end{align*}
\end{proof}
To guarantee the robust constraint satisfaction we have to take into
account not only an estimate $\hat{\theta}(t)$ of the vector of uncertain
parameters $\theta$, but the set of all admissible values $\hat{\Theta}(t)\subseteq\Theta$
(with $\theta,\hat{\theta}(t)\in\hat{\Theta}(t)$ for all $t\geq0$).
Following the result of Proposition \ref{prop:LS} we can calculate
an estimate for $\hat{\Theta}(t)$:
\begin{equation}
\hat{\Theta}(t)=\Theta\bigcap_{\tau\in[\ell,t]}\{\tilde{\theta}\in\Real^{d}:|\hat{\theta}(\tau)-\tilde{\theta}|\leq\Delta\theta(X)\}.\label{eq:set_LS}
\end{equation}
The property $\hat{\Theta}(t)\subseteq\Theta$
is satisfied for all $t\geq0$, and the size of $\hat{\Theta}(t)$
is shrinking. We also can use some updated estimates on $\Vert x\Vert_{\infty}$
in \eqref{eq:set_LS} instead of the worst case bound $X$.

\section{\label{sec:prediction} State Prediction}

 We aim to derive an interval predictor \cite{Dinh2014,leurent2019interval} for the system \eqref{eq:dynamics}, which takes the information on the observed current state $x(t)\in[y_{1}(t)-\underline{\nu}_{1}(t),y_{1}(t)+\overline{\nu}_{1}(t)]$, the estimated confidence region $\hat{\Theta}(t)$, a planned control signal $u:[t,+\infty)\to\Real^{q}$ and the admissible bounds
on the state perturbation $[\underline{\omega}(t),\overline{\omega}(t)]$; and outputs an interval $[\ux(t),\ox(t)]$ that must verify the inclusion property: 
\begin{equation}
\ux(s)\leq x(s)\leq\ox(s),\quad\forall s\geq t.\label{eq:inclusion-generic}
\end{equation}

There exist many predictors based, \emph{e.g}., on zonotope \cite{le2012}
or interval \cite{Dinh2014,leurent2019interval} representation of the
set of admissible values of $x(s)$.
Opting the simplicity of implementation and computational efficiency,
we use an interval predictor here that ensures the property \eqref{eq:inclusion-generic}.
To this end we will assume that the set $\hat{\Theta}(t)$ computed
by \eqref{eq:set_LS} is given, then there are two possible representations
of uncertainty of $A(\theta)$ in \eqref{eq:dynamics}:
\begin{itemize}
\item interval: for all $\theta\in\hat{\Theta}(t)$ and some $\underline{A},\overline{A}\in\Real^{p\times p}$,
\begin{equation}
\underline{A}\leq A(\theta)\leq\overline{A}\label{eq:interval}.
\end{equation}
\item polytopic: for all $\theta\in\hat{\Theta}(t)$ and some $A_{0}=A(\hat{\theta}(t))$
and $\Delta A_{i}=h_{i}\Delta\theta(X)$ for $h_{i}\in\{-1,1\}^{d}$
with $i\in[2d]$,
\begin{equation}
A(\theta)\in\left\{ A_{0}+\sum_{i=1}^{2^{d}}\alpha_{i}\Delta A_{i}:\alpha_{i}\geq0,\sum_{i=1}^{2^{d}}\alpha_{i}=1\right\} \label{eq:polytope}
\end{equation}
\end{itemize}
The matrices $\underline{A},\overline{A}$ can be calculated using
the interval arithmetic for $\hat{\Theta}(t)$; suggestions
for selection of $A_{0}$ and $\Delta A_{i}$ are
given in \eqref{eq:polytope} (other variants can be used \cite{delos2015}).

\subsection{Design of predictors}

A simple solution providing \eqref{eq:inclusion-generic} is proposed
in \cite{Efimov2012}, where the matrix interval arithmetic is used
to derive the predictor: 
\begin{proposition}
[Simple predictor of \cite{Efimov2012}] Assume that Assumption
\ref{assu:main} and the relations \eqref{eq:interval} are satisfied
for the system \eqref{eq:dynamics}. Then for $s\geq t$ the interval
predictor,
\begin{eqnarray}
\dot{\underline{x}}(s) & = & \underline{A}^{+}\underline{x}^{+}(s)-\overline{A}^{+}\underline{x}^{-}(s)-\underline{A}^{-}\overline{x}^{+}(s)+\overline{A}^{-}\overline{x}^{-}(s)\nonumber \\
 &  & +Bu(s)+D^{+}\underline{\omega}(s)-D^{-}\overline{\omega}(s),\label{eq:predictor-naive}\\
\dot{\overline{x}}(s) & = & \overline{A}^{+}\overline{x}^{+}(s)-\underline{A}^{+}\overline{x}^{-}(s)-\overline{A}^{-}\underline{x}^{+}(s)+\underline{A}^{-}\underline{x}^{-}(s)\nonumber \\
 &  & +Bu(s)+D^{+}\overline{\omega}(s)-D^{-}\underline{\omega}(s),\nonumber \\
 &  & \underline{x}(t)=y_{1}(t)-\underline{\nu}_{1}(t),\;\overline{x}(t)=y_{1}(t)+\overline{\nu}_{1}(t),\nonumber 
\end{eqnarray}
ensures the inclusion property \eqref{eq:inclusion-generic}.
\end{proposition}
However, \cite{leurent2019interval} showed that this predictor can
have unstable dynamics, even for stable systems, which causes a fast
explosion of the interval width $\overline{x}(s)-\underline{x}(s)$.
In that work, an enhanced predictor is proposed, which exploits the
polytopic structure \eqref{eq:polytope} to produce tighter and more
stable predictions, at the price of an additional requirement:
\begin{assumption}
\label{assu:metzler} There exists a nonsingular matrix $Z\in\Real^{p\times p}$
such that $Z^{-1}A_{0}Z$ is Metzler\footnote{We say that a matrix is Metzler when all its non-diagonal coefficients
are non-negative.}.
\end{assumption}
In practice, this assumption is often verified. It is for instance
the case whenever $A_{0}$ is diagonalizable, or a method from \cite{Efimov2013}
computes a similarity transformation $Z$ when the system is observable
with respect to a scalar output. To simplify the notation, we
further assume that $Z=I_{p}$. Denote $\Delta A_{+}=\sum_{i=1}^{2^{d}}\Delta A_{i}^{+}$
and $\Delta A_{-}=\sum_{i=1}^{2^{d}}\Delta A_{i}^{-}$.
\begin{proposition}
[Enhanced predictor of \cite{leurent2019interval}] \label{prop:predictor}
Assume that assumptions \ref{assu:main}, \ref{assu:metzler} and
the relation \eqref{eq:polytope} are satisfied for the system \eqref{eq:dynamics}.
Then for $s\geq t$ the interval predictor,
\begin{eqnarray}
\dot{\underline{x}}(s) & = & A_{0}\underline{x}(s)-\Delta A_{+}\underline{x}^{-}(s)-\Delta A_{-}\overline{x}^{+}(s)\nonumber \\
 &  & +Bu(s)+D^{+}\underline{\omega}(s)-D^{-}\overline{\omega}(s),\nonumber \\
\dot{\overline{x}}(s) & = & A_{0}\overline{x}(s)+\Delta A_{+}\overline{x}^{+}(s)+\Delta A_{-}\underline{x}^{-}(s)\label{eq:interval-predictor}\\
 &  & +Bu(s)+D^{+}\overline{\omega}(s)-D^{-}\underline{\omega}(s),\nonumber \\
 &  & \underline{x}(t)=y_{1}(t)-\underline{\nu}_{1}(t),\;\overline{x}(t)=y_{1}(t)+\overline{\nu}_{1}(t),\nonumber 
\end{eqnarray}
ensures the inclusion property \eqref{eq:inclusion-generic}.
\end{proposition}
In Fig. \ref{fig:predictor_example}, the difference in stability
of two predictors \eqref{eq:predictor-naive} and \eqref{eq:interval-predictor}
is illustrated for a simple example. \cite{leurent2019interval} suggest to always prefer
\eqref{eq:interval-predictor} whenever Assumption \ref{assu:metzler}
is verified, and only fallback to \eqref{eq:predictor-naive} as a
last resort. 
\begin{figure}
\begin{centering}
\includegraphics[width=1\linewidth]{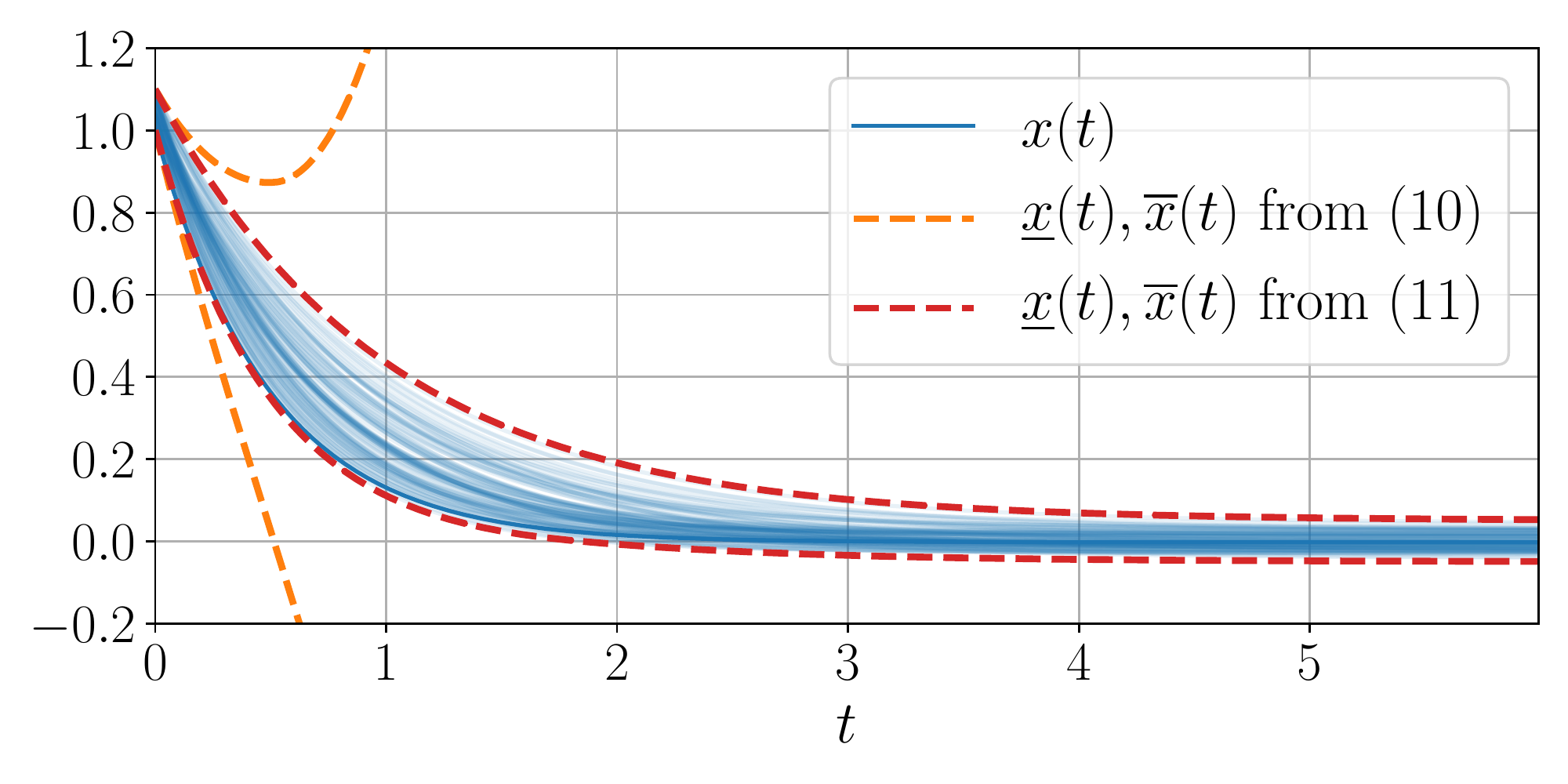}
\par\end{centering}
\caption{\label{fig:predictor_example} Comparison of \eqref{eq:predictor-naive}
and \eqref{eq:interval-predictor} for a simple system $\dot{x}(t)=-\theta x(t)+\omega(t)$,
with $\theta\in[1,2]$ and $\omega(t)\in[-0.05,0.05]$.}
\end{figure}

\subsection{Stabilizing control for \eqref{eq:predictor-naive} and \eqref{eq:interval-predictor}}

Note that both interval predictors, \eqref{eq:predictor-naive} and
\eqref{eq:interval-predictor}, admit a representation in the form (we use the
time argument $t=s$ in this subsection):
\begin{equation}
\dot{\xi}(t)=\cA_{0}\xi(t)+\cA_{1}\xi^{+}(t)+\cA_{2}\xi^{-}(t)+\cB u(t)+\delta(t),\label{eq:predictor}
\end{equation}
where $\xi(t)=[\underline{x}^{\top}(t)\;\overline{x}^{\top}(t)]^{\top}\in\Real^{2p}$
is the extended state vector of the predictors,
\[
\delta(t)=\left[\begin{array}{cc}
D^{+} & -D^{-}\\
-D^{-} & D^{+}
\end{array}\right]\left[\begin{array}{c}
\underline{\omega}(t)\\
\overline{\omega}(t)
\end{array}\right]\in\Real^{2p}
\]
is the external known input, $\cB=[B^{\top}\;B^{\top}]^{\top}$, 
\[
\cA_{0}=0,\;\cA_{1}=\left[\begin{array}{cc}
\underline{A}^{+} & -\underline{A}^{-}\\
-\overline{A}^{-} & \overline{A}^{+}
\end{array}\right],\;\cA_{2}=\left[\begin{array}{cc}
-\overline{A}^{+} & \overline{A}^{-}\\
\underline{A}^{-} & -\underline{A}^{+}
\end{array}\right]
\]
for \eqref{eq:predictor-naive} and
\[
\cA_{0}=\left[\begin{array}{cc}
A_{0} & 0\\
0 & A_{0}
\end{array}\right],\;\cA_{1}=\left[\begin{array}{cc}
0 & -\Delta A_{-}\\
0 & \Delta A_{+}
\end{array}\right],\;\cA_{2}=\left[\begin{array}{cc}
-\Delta A_{+} & 0\\
\Delta A_{-} & 0
\end{array}\right]
\]
for \eqref{eq:interval-predictor}. Note that \eqref{eq:predictor}
is a nonlinear system due to the presence of globally Lipschitz nonlinearities
$\xi^{+}(t)$ and $\xi^{-}(t)$. 

Due to \eqref{eq:inclusion-generic}, the boundedness of $\xi(t)$
implies the same property of $x(t)$. Therefore, in order to regulate
\eqref{eq:dynamics} it is required to design a state feedback $u(t)$
minimizing the asymptotic amplitude of the state $\xi(t)$ for given
input $\delta(t)$ \cite{Efimov2013a}. In other words, it is necessary
to design a control $u(t)$ that input-to-state stabilizes \eqref{eq:predictor}.
It is proposed to look for such a control in the form:
\begin{equation}
u(t)=K_{0}\xi(t)+K_{1}\xi^{+}(t)+K_{2}\xi^{-}(t)+S\delta(t),\label{eq:control_pr}
\end{equation}
where $K_{0},K_{1},K_{2}\in\Real^{q\times2p}$ and $S\in\Real^{q\times2p}$
are the gains to be designed (\eqref{eq:control_pr} contains a nonlinear
feedback). The selection of $S$ is simple, it has to minimize the
norm of $\cB S+I_{2p}$, and it can be made independently of $K_{0},K_{1},K_{2}$.
Therefore, denoting $\tilde{\delta}(t)=(\cB S+I_{2p})\delta(t)$ the
closed-loop system \eqref{eq:predictor}, \eqref{eq:control_pr} takes
the form:
\begin{equation}
\dot{\xi}(t)=\cD_{0}\xi(t)+\cD_{1}\xi^{+}(t)+\cD_{2}\xi^{-}(t)+\tilde{\delta}(t),\label{eq:closed-loop_pr}
\end{equation}
where $\cD_{i}=\cA_{i}+\cB K_{i}$ for $i\in[3]$, and the restrictions,
which the gains $K_{0},K_{1},K_{2}$ have to respect, are given below:
\begin{theorem}
\label{thm:ISS_pr} If there exist diagonal matrices $P$, $Q$, $Q_{+}$,
$Q_{-}$, $Z_{+}$, $Z_{-}$, $\Psi_{+}$, $\Psi_{-}$, $\Psi$, $\Gamma\in\Real^{2p\times2p}$
such that the following linear matrix inequalities are satisfied:
\begin{gather*}
P+\min\{Z_{+},Z_{-}\}>0,\;\Upsilon\preceq0,\;\Gamma>0,\\
Q+\min\{Q_{+},Q_{-}\}+2\min\{\Psi_{+},\Psi_{-}\}>0,\\
\text{where }\quad \Upsilon=\left[\begin{array}{cccc}
\Upsilon_{11} & \Upsilon_{12} & \Upsilon_{13} & P\\
\Upsilon_{12}^{\top} & \Upsilon_{22} & \Upsilon_{23} & Z_{+}\\
\Upsilon_{13}^{\top} & \Upsilon_{23}^{\top} & \Upsilon_{33} & -Z_{-}\\
P & Z_{+} & -Z_{-} & -\Gamma
\end{array}\right],\\
\Upsilon_{11}=\cD_{0}^{\top}P+P\cD_{0}+Q,\;\Upsilon_{12}=\cD_{0}^{\top}Z_{+}+P\cD_{1}+\Psi_{+},\\
\Upsilon_{13}=P\cD_{2}-\cD_{0}^{\top}Z_{-}-\Psi_{-},\;\Upsilon_{22}=Z_{+}\cD_{1}+\cD_{1}^{\top}Z_{+}+Q_{+},\\
\Upsilon_{23}=Z_{+}\cD_{2}-\cD_{1}^{\top}Z_{-}+\Psi,\;\Upsilon_{33}=-Z_{-}\cD_{2}-\cD_{2}^{\top}Z_{-}+Q_{-},
\end{gather*}
then \eqref{eq:closed-loop_pr} is input-to-state stable with respect
to $\underline{\omega},\overline{\omega}$.
\end{theorem}
Note that the requirement that $P$ has to be diagonal is not
restrictive, since for a Metzler matrix $\cD_{0}$ (the case of \eqref{eq:predictor-naive}
and \eqref{eq:interval-predictor}), its stability is equivalent to
existence of a diagonal solution $P$ of the Lyapunov equation $\cD_{0}^{\top}P+P\cD_{0}\prec0$
\cite{Positive}.
\begin{proof}
Consider a candidate Lyapunov function:
\begin{gather*}
V(\xi)=\xi^{\top}P\xi+\xi{}^{\top}Z_{+}\xi^{+}-\xi^{\top}Z_{-}\xi^{-}\\
=\sum_{k=1}^{2p}P_{k,k}\xi_{k}^{2}+(Z_{+})_{k,k}|\xi_{k}|\xi_{k}^{+}+(Z_{-})_{k,k}|\xi_{k}|\xi_{k}^{-},
\end{gather*}
which is positive definite provided that
$
P+\min\{Z_{+},Z_{-}\}>0
$
since all terms in $V$ are quadratic-like, and whose derivative for
the system \eqref{eq:closed-loop_pr} dynamics takes the form:
\begin{gather*}
\dot{V}=2\dot{\xi}^{\top}P\xi+2\dot{\xi}^{\top}Z_{+}\xi^{+}-2\dot{\xi}^{\top}Z_{-}\xi^{-}\\
=\left[\begin{array}{c}
\xi\\
\xi^{+}\\
\xi^{-}\\
\tilde{\delta}
\end{array}\right]^{\top}\Upsilon\left[\begin{array}{c}
\xi\\
\xi^{+}\\
\xi^{-}\\
\tilde{\delta}
\end{array}\right]-\xi^{\top}Q\xi-(\xi^{+})^{\top}Q_{+}\xi^{+}\\
-(\xi^{-})^{\top}Q_{-}\xi^{-}-2(\xi^{+})^{\top}\Psi\xi^{-}-2(\xi^{+})^{\top}\Psi_{+}\xi\\
-2(-\xi^{-})^{\top}\Psi_{-}\xi+\tilde{\delta}^{\top}\Gamma\tilde{\delta}.
\end{gather*}
Note that
$
(\xi^{+})^{\top}\Psi\xi^{-}=0,\;(\xi^{+})^{\top}\Psi_{+}\xi\geq0,\;(-\xi^{-})^{\top}\Psi_{-}\xi\geq0
$
for any diagonal matrix $\Psi$ and $\Psi_{+}\geq0$, $\Psi_{-}\geq0$.
Hence, if $\Upsilon\preceq0$, as it is assumed in the theorem, we
obtain that
\begin{eqnarray*}
\dot{V} & \leq & -\xi^{\top}Q\xi-(\xi^{+})^{\top}Q_{+}\xi^{+}-(\xi^{-})^{\top}Q_{-}\xi^{-}\\
 &  & -2(\xi^{+})^{\top}\Psi_{+}\xi-2(-\xi^{-})^{\top}\Psi_{-}\xi+\tilde{\delta}^{\top}\Gamma\tilde{\delta}\\
 & \leq & -\xi^{\top}\Omega\xi+\tilde{\delta}^{\top}\Gamma\tilde{\delta},
\end{eqnarray*}
where $\Omega=Q+\min\{Q_{+},Q_{-}\}+2\min\{\Psi_{+},\Psi_{-}\}>0
$
is a diagonal matrix. The substantiated properties of $V$ and its
derivative imply that \eqref{eq:closed-loop_pr} is input-to-state
stable \cite{Sontag:01:Springer,Dashkovskiy:11:AiT} with respect
to the input $\tilde{\delta}$ (or, by its definition,
to $(\underline{\omega},\overline{\omega})$).
\end{proof}
Following the proof of Theorem \ref{thm:ISS_pr}, for all $\xi\in\Real^{2p}$,
\[
\xi^{\top}(P+\min\{Z_{+},Z_{-}\})\xi\leq V(\xi)\leq\xi^{\top}(P+Z_{+}^{+}+Z_{-}^{+})\xi,
\]
then
$
\dot{V}\leq-\alpha V+\tilde{\delta}^{\top}\Gamma\tilde{\delta}
$
for all $\xi,\tilde{\delta}\in\Real^{2p}$, where
$
\alpha=\min_{i\in[2p]}\lambda_{i}\left(\Omega(P+Z_{+}^{+}+Z_{-}^{+})^{-1}\right),
$
and we can define the set (recall that the signal $\tilde{\delta}(t)$
is known for all $t\geq0$)
\begin{equation}
\X_{f}=\{\xi\in\Real^{2p}:V(\xi)\leq\alpha^{-1}\sup_{t\geq0}|\tilde{\delta}^{\top}(t)\Gamma\tilde{\delta}(t)|\},\label{eq:X_f}
\end{equation}
as the set that asymptotically attracts all trajectories in \eqref{eq:closed-loop_pr}.

The conditions of Theorem \ref{thm:ISS_pr} assume that the control
gains $K_{0},K_{1},K_{2}$ are given, let us find these gains as solutions
of linear matrix inequalities:
\begin{corollary}
If there exist diagonal matrices $P$, $\tilde{Q}$, $\tilde{Q}_{+}$,
$\tilde{Q}_{-}$, $Z_{+}$, $Z_{-}$, $\tilde{\Psi}_{+}$, $\tilde{\Psi}_{-}$,
$\tilde{\Psi}$, $\Gamma\in\Real^{2p\times2p}$ and matrices $U_{0},U_{1},U_{2}\in\Real^{q\times2p}$
satisfying following linear matrix inequalities:
\begin{gather*}
P>0,\;Z_{+}>0,\;Z_{-}>0,\;\Pi\preceq0,\;\Gamma>0,\\
\tilde{Q}+\min\{\tilde{Q}_{+},\tilde{Q}_{-}\}+2\min\{\tilde{\Psi}_{+},\tilde{\Psi}_{-}\}>0,\\
\text{where }\quad \Pi=\left[\begin{array}{cccc}
\Pi_{11} & \Pi_{12} & \Pi_{13} & I\\
\Pi_{12}^{\top} & \Pi_{22} & \Pi_{23} & I\\
\Pi_{13}^{\top} & \Pi_{23}^{\top} & \Pi_{33} & -I\\
I & I & -I & -\Gamma
\end{array}\right],\\
\Pi_{11}=P^{-1}\cA_{0}^{\top}+\cA_{0}P^{-1}+U_{0}^{\top}\cB^{\top}+\cB U_{0}+\tilde{Q},\\
\Pi_{12}=\cA_{1}Z_{+}^{-1}+\cB U_{1}+P^{-1}\cA_{0}^{\top}+U_{0}^{\top}\cB^{\top}+\tilde{\Psi}_{+},\\
\Pi_{13}=\cA_{2}Z_{-}^{-1}+\cB U_{2}-P^{-1}\cA_{0}^{\top}-U_{0}^{\top}\cB^{\top}-\tilde{\Psi}_{-},\\
\Pi_{22}=Z_{+}^{-1}\cA_{1}^{\top}+\cA_{1}Z_{+}^{-1}+U_{1}^{\top}\cB^{\top}+\cB U_{1}+\tilde{Q}_{+},\\
\Pi_{23}=\cA_{2}Z_{-}^{-1}+\cB U_{2}-Z_{+}^{-1}\cA_{1}^{\top}-U_{1}^{\top}\cB^{\top}+\tilde{\Psi},\\
\Pi_{33}=\tilde{Q}_{-}-Z_{-}^{-1}\cA_{2}^{\top}-\cA_{2}Z_{-}^{-1}-U_{2}^{\top}\cB^{\top}-\cB U_{2},
\end{gather*}
then \eqref{eq:closed-loop_pr} for $K_{0}=U_{0}P$, $K_{1}=U_{1}Z_{+}$
and $K_{2}=U_{2}Z_{-}$ is input-to-state stable with respect to the
inputs $\underline{\omega},\overline{\omega}$.
\end{corollary}
\begin{proof}
Note that the conditions $P>0$, $Z_{+}>0$, $Z_{-}>0$ imply $P+\min\{Z_{+},Z_{-}\}>0$,
and 
\[
\Upsilon=\left[\begin{array}{cccc}
P & 0 & 0 & 0\\
0 & Z_{+} & 0 & 0\\
0 & 0 & Z_{-} & 0\\
0 & 0 & 0 & I_{2p}
\end{array}\right]\Pi\left[\begin{array}{cccc}
P & 0 & 0 & 0\\
0 & Z_{+} & 0 & 0\\
0 & 0 & Z_{-} & 0\\
0 & 0 & 0 & I_{2p}
\end{array}\right]
\]
under substitution $U_{0}=K_{0}P^{-1}$, $U_{1}=K_{1}Z_{+}^{-1}$,
$U_{2}=K_{2}Z_{-}^{-1}$, $\tilde{Q}=P^{-1}QP^{-1}$, $\tilde{Q}_{+}=Z_{+}^{-1}Q_{+}Z_{+}^{-1}$,
$\tilde{Q}_{-}=Z_{-}^{-1}Q_{-}Z_{-}^{-1}$, $\tilde{\Psi}=Z_{-}^{-1}\Psi Z_{+}^{-1}$,
$\tilde{\Psi}_{+}=P^{-1}\Psi_{+}Z_{+}^{-1}$ and $\tilde{\Psi}_{-}=P^{-1}\Psi_{-}Z_{-}^{-1}$.
Hence, $\Upsilon\preceq0$ provided that $\Pi\preceq0$. The inequalities
$\tilde{Q}+\min\{\tilde{Q}_{+},\tilde{Q}_{-}\}+2\min\{\tilde{\Psi}_{+},\tilde{\Psi}_{-}\}>0$
and $Q+\min\{Q_{+},Q_{-}\}+2\min\{\Psi_{+},\Psi_{-}\}>0$ are equivalent
due to the diagonal structure of all matrices. Therefore, under introduced
restrictions all conditions of Theorem \ref{thm:ISS_pr} are verified
for $K_{0}=U_{0}P$, $K_{1}=U_{1}Z_{+}$ and $K_{2}=U_{2}Z_{-}$.
\end{proof}
The requirements imposed on $P,Z_{+},Z_{-}$ in this corollary are
more restrictive than the conditions of Theorem \ref{thm:ISS_pr},
but it allows the gains $K_{0},K_{1},K_{2}$ to be efficiently calculated.

Under conditions of Theorem \ref{thm:ISS_pr}, the control \eqref{eq:control_pr}
ensures stabilization of the predictor \eqref{eq:predictor} in
a vicinity $\X_{f}$ of the origin whose size is proportional
to the system \eqref{eq:dynamics} uncertainty (it can be optimized
by the choice of $K_{0},K_{1},K_{2}$). Due to \eqref{eq:inclusion-generic},
the system \eqref{eq:dynamics} also will reach a
neighborhood of the origin and the posed control problem would be solved provided that \eqref{eq:constraints}
holds. In order to ensure the robust constraint satisfaction we
consider an MPC design in the next section.

\section{\label{sec:control} Robust Control}

For brevity, the results of this section are given for the predictor
\eqref{eq:interval-predictor} only. We need the following hypothesis:
\begin{assumption}
\label{assu:ctrl} There exist $K_{0},K_{1},K_{2}\in\Real^{q\times2p}$
satisfying the conditions of Theorem \ref{thm:ISS_pr} for the matrices
$A_{0}$ and $\Delta A_{i}$ with $i\in[2d]$ calculated in \eqref{eq:polytope}
for $\hat{\Theta}(t)=\Theta$, and
$
\X_{f}\subset\X^{2},
$
where the corresponding set $\X_{f}$ is given in \eqref{eq:X_f},
and
\[
K_{0}\xi+K_{1}\xi^{+}+K_{2}\xi^{-}+S\delta(t)\in\U
\]
for any $\xi\in\X_{f}$ and $t\geq0$.
\end{assumption}
These properties guarantee that there exists a control \eqref{eq:control_pr}
that can be always applied to stabilize the predictor \eqref{eq:interval-predictor} and into the set $\X_{f}$ the restrictions \eqref{eq:constraints}
also hold for such a control. Define $T>0$ and $\tau\in(0,T)$ as the interval of prediction and
the application time for MPC. 
Denote $t_{i}=i\tau$ for $i\in\Natural_{+}$, then the developed MPC algorithm can be formalized as:
\begin{enumerate}
\item Take $\hat{\Theta}(t_{i})$ from \eqref{eq:set_LS} and calculate
the matrices $A_{0}$ and $\Delta A_{i}$ with $i\in[2d]$ for \eqref{eq:polytope}.
\item Find, given weights $W_{i}\succeq0$ in $\Real^{2p\times2p}$, the controls
\begin{gather}
\cU=\text{argmin}_{u:[t_{i},t_{i}+T]\to\Real^{q}}\xi^{\top}(t_{i}+T)W_{1}\xi(t_{i}+T)\nonumber\\
+\int_{t_{i}}^{t_{i}+T}\xi^{\top}(s)W_{2}\xi(s)+u^{\top}(s)W_{3}u(s)ds,\label{eq:OCP}
\end{gather}
such that the following constraints are satisfied: 
\begin{enumerate}
\item $\xi:[t_{i},t_{i}+T]\to\Real^{2p}$ is a solution of \eqref{eq:interval-predictor}
\item $\xi(s)\in\X^{2}$ and $u(s)\in\U$ for $s\in[t_{i},t_{i}+T]$; 
\item $\xi(t_{i}+T)\in\X_{f}$.
\end{enumerate}
\item For $t\in[t_{i},t_{i}+\tau)$ select
\begin{equation}
u(t)=\begin{cases}
\cU(t) & \xi(t_{i})\notin\X_{f}\\
\eqref{eq:control_pr} & \xi(t_{i})\in\X_{f}
\end{cases},\label{eq:control}
\end{equation}
where $K_{0},K_{1},K_{2}$ are taken from Assumption \ref{assu:ctrl}.
\end{enumerate}
As we can conclude, the idea of the proposed dual MPC scheme (see
also \cite{Michalska1993,MPC1,MPC:Tube2}) is to use an open-loop
optimal control to reach a neighborhood of the origin $\X_{f}$ ensuring
a robust constraint satisfaction \eqref{eq:constraints}, where a
closed-loop control \eqref{eq:control_pr} can be applied, which provides
asymptotic performances (stability and robustness, also with constraint
satisfaction due to Assumption \ref{assu:ctrl} and the definition
of the terminal set \eqref{eq:X_f}). 


The main result of the paper is as follows:
\begin{theorem}
\label{th:MPC} Let $\underline{x}_{0},\overline{x}_{0}\in\X$, and
assumptions \ref{assu:main}--\ref{assu:ctrl} hold with $\overline{\omega},\overline{\omega}-\underline{\omega}$
being non-increasing functions of $t\geq0$. Then the closed-loop
system given by \eqref{eq:dynamics}, \eqref{eq:outputs}, \eqref{eq:interval-predictor}
and \eqref{eq:control} has the following properties:
\begin{enumerate}
\item Input-to-state stability for $\underline{x},\;\overline{x}$ and practical
input-to-state stability for $x$ with respect to $\underline{\omega},\overline{\omega}$
in the terminal set $\X_{f}$; 
\item Recursive feasibility with reaching $\X_{f}$ in a finite time; 
\item Constraint satisfaction.
\end{enumerate}
\end{theorem}
\begin{proof}
Recall that $\theta\in\hat{\Theta}(t)$ for all $t\geq0$ due to the
result of Proposition \ref{prop:LS}, and the size of the set $\hat{\Theta}(t)$
is not growing with time by definition of \eqref{eq:set_LS}.

\emph{1)} Note that if for some $t_{k}\geq0$ the initial conditions $(\underline{x}^{\top}(t_{k}),\overline{x}^{\top}(t_{k}))^{\top}\in\X_{f}\subset\X^{2}$,
then the control \eqref{eq:control} equals to \eqref{eq:control_pr}.
According to the definition \eqref{eq:X_f} of $\X_{f}$ and Assumption
\ref{assu:ctrl}, $\xi(t)\in\X_{f}$ and $u(t)\in\U$ for all $t\geq t_{k}$,
and the system is input-to-state stable with respect to $\xi(t)=[\underline{x}^{\top}(t)\;\overline{x}^{\top}(t)]^{\top}$
due to the result of Theorem \ref{thm:ISS_pr}. Since $|x(t)|\leq|\xi(t)|$
under \eqref{eq:inclusion-generic} for $t\geq t_{k}$ and $|\xi(t_{k})|\leq|x(t_{k})|+\zeta$
with $\zeta>0$ 
, the practical input-to-state stability for the variable

\emph{2)} Now, let $(\underline{x}^{\top}(0),\overline{x}^{\top}(0))^{\top}\in\X^{2}\setminus\X_{f}$
and assume that there is a solution of the optimal control problem
\eqref{eq:OCP}. Applying such a control through \eqref{eq:control}
for $t\in[0,\tau)$, we have that $\xi(t)\in\X$ and $u(t)\in\U$
on this time interval. At $t=t_{1}=\tau$, if again $(\underline{x}^{\top}(t_{1}),\overline{x}^{\top}(t_{1}))^{\top}\in\X^{2}\setminus\X_{f}$
, then it recursively exists a solution to
\eqref{eq:OCP} since the set $\hat{\Theta}(t)$ is shrinking by its
design and the signals $\overline{\omega}(t),\overline{\omega}(t)-\underline{\omega}(t)$
are non-increasing by hypotheses of the theorem (\emph{i.e}., the
solution obtained at $t_{i}$ is a sub-optimal branch of the solution
calculated at $t_{i-1}$ for all $i\geq1$). Thus, recursive feasibility
follows. Note that $\X_{f}$ is a neighborhood of the origin, and
the given in \eqref{eq:OCP} cost with positive definite matrices
$W_{1}$, $W_{2}$ and $W_{3}$ is minimized inside $\X_{f}$. Using
this and sub-optimality arguments, since $\xi(t_{i}+T)\in\X_{f}$
in \eqref{eq:OCP} (provided that the optimal control $\cU$ is applied)
and $[\underline{x}(t_{i}),\overline{x}(t_{i})]\subset[\underline{x}(t_{i-1}+\tau),\overline{x}(t_{i-1}+\tau)]$
for all $i\geq1$, there is a finite time instant $t_{k}\geq T$ such
that $(\underline{x}^{\top}(t_{k}),\overline{x}^{\top}(t_{k}))^{\top}\in\X_{f}$,
and the system further stays there.

\emph{3)} is a consequence of the previous analysis: under
the control \eqref{eq:control} the constrains \eqref{eq:constraints}
are always satisfied.
\end{proof}

\section{\label{sec:experiments} Numerical experiment}

We tackle the problem of the robust adaptive lateral control of an autonomous vehicle with unknown tire friction, for a lane-keeping application. We represent the state of a rigid vehicle by its position $(p_x, p_y)$, angle $\psi$, velocity $(v_x, v_y)$ in the body frame and yaw rate $r$. We denote its mass as $m$, moment of inertia as $I_z$, and front and rear axle positions as $a,b$. We consider the Dynamical Bicycle Model described in Chapter 3.2 of \cite{awan2014compensation} parametrised by the unknown front and rear tire friction coefficients $\theta = \begin{bmatrix} C_{\alpha_f} & C_{\alpha_r}\end{bmatrix}^\transp$, which yields the linear dynamics \eqref{eq:dynamics} with 
\[
x = \begin{bmatrix} {p_y} \\ {\psi} \\ {v_y} \\ {r} \end{bmatrix},\quad
A = \begin{bmatrix}
0 & v_x & 1 & 0 \\
0 & 0 & 0 & 1 \\
0 & 0 & 0 & - v_x \\
0 & 0 & 0 & 0
\end{bmatrix},\quad
B =
\begin{bmatrix}
0 \\
0 \\
\frac{2}{m} \\
\frac{a}{I_z}
\end{bmatrix},
\]
\[
\phi = \frac{-2}{m v_x I_z}\left[\begin{bmatrix}
0 & 0 & 0 & 0 \\
0 & 0 & 0 & 0 \\
0 & 0 & I_z & a I_z \\
0 & 0 & a m & a^2 m \\
\end{bmatrix},\begin{bmatrix}
0 & 0 & 0 & 0 \\
0 & 0 & 0 & 0 \\
0 & 0 & I_z & -b I_z \\
0 & 0 & - bm & b^2 m \\
\end{bmatrix}\right].
\]
Instead of simply stabilizing the vehicle state $x$, we track the lateral position $y_r(t)$ of the lane center. However, we do not have access to a full state reference $x_r(t) = [y_r(t), \psi_r(t), v_{y,r}(t), r_r(t)]^\transp$ consistent with the dynamics \eqref{eq:dynamics}. Thus, we define the state $\tilde{x} = x - [y_r(t), 0, 0, 0]^\transp$ and consider the remaining unknown terms $[0, \psi_r(t), v_{y,r}(t), r_r(t)]$ and $u_r(t)$ as perturbations $\omega(t)$, bounded since $x_r,\,u_r$ are assumed to belong to $\X = \pm[3, 2, 6, 6]^\transp$ and $\U=\pm[10]$.

The \Cref{fig:lane-keeping} depicts our approach. The confidence region $\hat{\Theta}(t)$ from \eqref{eq:set_LS} is shown in the top graph, and shrinks with time. To simplify verification of \Cref{assu:metzler} for this example, an auxiliary preliminary feedback has been applied shifting the eigenvalues of the closed-loop system. The robust stability of this feedback is assessed with the LMI of \Cref{thm:ISS_pr}, and we compute the corresponding basin of attraction $X_f$ from \eqref{eq:X_f}, represented in green in the bottom subfigure. Then, we use a sampling-based MPC scheme \cite{HomemDeMello2014} to solve \eqref{eq:LS} and bring $\xi(t)$ into $X_f$ in $T=3s$.  The associated interval prediction $\xi(t)$ from \eqref{eq:predictor} is represented with a color gradient from $t=t_i$ (red) to $t=t_i+T$ (green). Once the vehicles enters $X_f$, we finally switch to the closed-loop feedback \eqref{eq:control_pr} following \eqref{eq:control} for the rest of the simulation. A video is available at \href{https://youtu.be/axurBzHRLGY}{this url}.
\begin{figure}
    \centering
    \includegraphics[width=\linewidth]{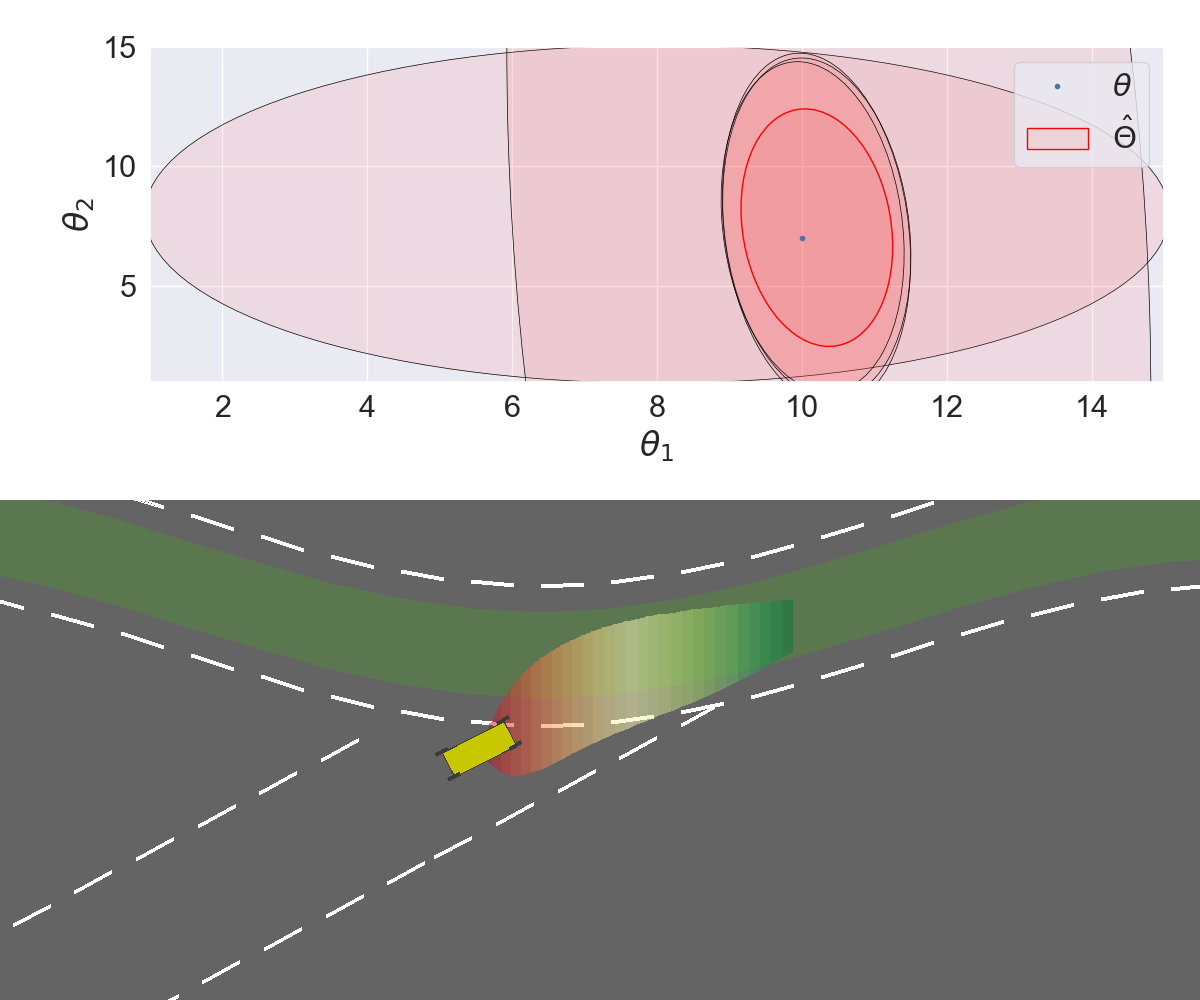}
    \caption{\textbf{Top}: the model estimation showing the confidence region $\hat{\Theta}(t)$ from \eqref{eq:predictor} at different times $t$. \textbf{Bottom}: a lane keeping application, where a car must follow a lane-center curve under unknown friction and perturbations. $X_f$ is shown in green, and $\xi(t)$ as an area with a color gradient.}
    \label{fig:lane-keeping}
\end{figure}

\section*{Conclusion}

A robust adaptive MPC algorithm is presented for a partially known
linear system subject to disturbances. The peculiarity of the proposed
solution consists in utilization of interval predictor. The applicability
of the method is demonstrated on a simulated car application.

\bibliographystyle{IEEEtran}
\bibliography{references}

\end{document}